\documentclass[conference]{IEEEtran}

\usepackage{amsthm,amssymb,graphicx,multirow,amsmath,amsfonts,cite}
\usepackage[usenames,dvipsnames]{color}
\usepackage{tikz}
\usetikzlibrary{arrows}

 \DeclareMathOperator{\tr}{tr}

\newsavebox{\mybox}

\newtheorem{lemma}{Lemma}
\newtheorem{theorem}{Theorem}

\ifCLASSINFOpdf
\else
\fi

\begin{document}
\title{On the Capacity of the Half-Duplex MIMO Gaussian Diamond Channel}

\author{\IEEEauthorblockN{Antony V. Mampilly and Srikrishna Bhashyam}
\IEEEauthorblockA{Department of Electrical Engineering\\
Indian Institute of Technology Madras\\
Chennai 600036, India.}

}


\maketitle

\begin{abstract}
In this paper, we analyze the 2-relay multiple-input multiple-output (MIMO) Gaussian diamond channel. We show that a multihopping decode-and-forward with multiple access (MDF-MAC)  protocol achieves rates within a constant gap from capacity when a channel parameter $\Delta$ is greater than zero. We also identify the transmit covariance matrices to be used by each relay in the multiple-access (MAC) state of the MDF-MAC protocol. As done for the single-antenna 2-relay Gaussian diamond channel, the channel parameter $\Delta$ is defined to be the difference between the product of the capacities of the links from the source to the two relays and the product of the capacities of the links from the two relays to the destination. 
\end{abstract}

\section{Introduction}
The relay channel was introduced in \cite{van71,van71-2} and studied in \cite{CovElG79}. Although the relay channel has been studied extensively, the exact capacity of the channel is still unknown. The approximate capacity of a single-antenna Gaussian relay channel to within one bit was found in \cite{ChaChuLee10}. The multi-input multi-output (MIMO) Gaussian relay channel was studied in \cite{WanZanHos05} and the approximate capacity of the MIMO Gaussian relay channel to within a finite number of bits was recently found in \cite{JinKim14,JinKim15}.

The parallel relay channel or diamond channel was introduced in \cite{SchGal00}. This channel consists of a source, a destination and $N$ relays. The relays cannot communicate with each other and the source cannot directly communicate with the destination. The full-duplex (FD) single-antenna Gaussian $N$-relay diamond channel was studied in \cite{SenWanFra12,CheOzg14} and the capacity obtained to within $O(\log N)$ bits. Noisy network coding \cite{LimKimElGChu11}, which is applicable to more general relay network topologies, achieves a gap of $O(N)$ bits for the Gaussian $N$-relay diamond channel. The half-duplex (HD) single-antenna Gaussian $N$-relay diamond channel has been studied in \cite{BraOzgFra12,BraFra14,CarTunKnoSal14,BagMotKha14}. The HD $N$-relay diamond channel can be in $2^N$ relaying states since each relay can be in transmit or receive state at any time. In \cite{BraOzgFra12,BraFra14}, it was proved, for $N \le 6$, that the optimal relaying protocol has at most $N + 1$ states, and the same was conjectured for general $N$. In \cite{CarTunKnoSal14}, noisy network coding was shown to achieve rates within $1.96(N + 2)$ bits of the cut-set upper bound. In \cite{BagMotKha14}, simple multi-hopping decode-and-forward (MDF) protocols are proposed and shown to achieve capacity within 0.71 bits for the case of $N = 2$ with fixed scheduling and constant power constraints across all relaying states. An important parameter of the 2-relay diamond channel $\Delta$ was also introduced in \cite{BagMotKha14}. Intuitively, $\Delta$ is a measure comparing the capacities of the links in the first hop with the links in the second hop. All the possible channel conditions, namely $\Delta = 0$, $\Delta > 0$ and $\Delta < 0$, were analyzed.  For $\Delta = 0$, the MDF protocol achieved exact capacity. For $\Delta > 0$ and $\Delta < 0$, MDF-MAC and MDF-BC protocols were proposed and shown to achieve rates within 0.71 bits of capacity. 

The HD MIMO parallel relay channel or HD MIMO diamond channel has been recently studied in \cite{CarTunKno15}. In \cite{CarTunKno15}, (1) noisy network coding was shown to achieve rates within  $1.96(N + 2)$ bits {\em per antenna} of the cut-set upper bound for the HD MIMO diamond channel, and (2) it was also shown that this finite gap can be achieved using at most $N + 1$ relaying states. 
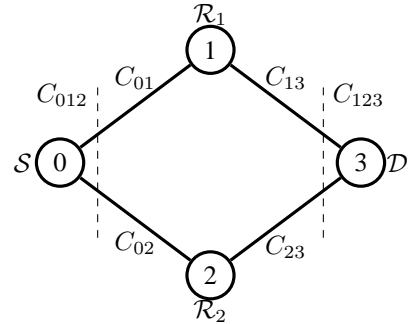
\begin{figure}[ht]
\begin{center}
\begin{tikzpicture}[scale=1]
\tikzset{vertex/.style = {shape=circle,very thick,draw,minimum size=1.1em}}
\tikzset{edge/.style = {> = latex', very thick}}

\node[vertex] (0) at (0,0) {0};
\node[vertex] (1) at (2,1.5) {1};
\node[vertex] (2) at (2,-1.5) {2};
\node[vertex] (3) at (4,0) {3};

\draw [edge] (0) to (1);
\draw [edge] (0) to (2);
\draw [edge] (1) to (3);
\draw [edge] (2) to (3);

\node[left] at (-0.25,0) {${\cal S}$ };
\node[right] at (4.22,0) {${\cal D}$};
\node[above] at (2,1.7) {${\cal R}_1$};
\node[below] at (2,-1.7) {${\cal R}_2$};

\node[above] at (1,0.85) {$C_{01}$};
\node[below] at (1,-0.8) {$C_{02}$};
\node[above] at (3,0.85) {$C_{13}$};
\node[below] at (3,-0.85) {$C_{23}$};
\node[right] at (3.5,0.9) {$C_{123}$};
\node[left] at (0.5,0.9) {$C_{012}$};

\draw[dashed] (0.5,1)-- (0.5,-1);
\draw[dashed] (3.5,1) -- (3.5,-1);
\end{tikzpicture}
\caption{2-relay MIMO Gaussian diamond channel: Each node has $n$ antennas, Parameter $\Delta = C_{01}C_{02} - C_{13}C_{23}$, $C_{012}$ and $C_{123}$ are the cut capacities of the respective cuts}
\label{MIMOdiamond}
\end{center}
\end{figure}

In this paper, we consider the 2-relay MIMO Gaussian diamond channel (see Fig. \ref{MIMOdiamond}), i.e., the multi-antenna generalization of the diamond channel considered in \cite{BagMotKha14}. We show that, for $\Delta > 0$, the multihopping decode-and-forward with multiple access (MDF-MAC) protocol achieves rates within a constant gap from capacity. In the process, we also identify the transmit covariance matrices to be used by each relay in the multiple-access (MAC) state. In the single-antenna case in \cite{BagMotKha14}, the achievable rate of the MDF-MAC protocol is the solution to a linear program since the rate constraints in the MAC state are three linear constraints. However, the rate region for the MAC state in the multi-antenna setting in this paper is the union of the rate regions for each choice of feasible transmit covariance matrices, i.e., an infinite union. We choose the transmit covariance matrices in the MAC state appropriately to restrict ourselves to a region described by three linear rate constraints {\em and} obtain a finite gap between the achievable rate and the capacity upper bound.  For $\Delta = 0$, the capacity has already been determined in \cite{BagMotKha14}. For $\Delta <0$, it is not yet known if the MDF-BC protocol can achieve rates within a constant gap from capacity as in the single-antenna case.   

\section{System Model and MDF-MAC protocol}
The 2-relay MIMO Gaussian diamond channel is shown in Fig. \ref{MIMOdiamond}.
For ease of exposition, we assume that all nodes have $n$ antennas. The received signals at relays ${\cal R}_1$, ${\cal R}_2$ and destination ${\cal D}$ are given by:
\begin{IEEEeqnarray}{rll}
\mathbf{y}_1 &=& {\mathbf H}_{01} \mathbf{x}_0 + \mathbf{z}_1,  \nonumber\\
\mathbf{y}_2 &=& {\mathbf H}_{02} \mathbf{x}_0 + \mathbf{z}_2,  \nonumber\\
\mathbf{y}\; &=& {\mathbf H}_{13} \mathbf{x}_1 + {\mathbf H}_{23} \mathbf{x}_2 + \mathbf{z}_3,
\end{IEEEeqnarray}
respectively, where $\mathbf{x}_0$, $\mathbf{x}_1$, and $\mathbf{x}_2$ are the transmit signals from ${\cal S}$, ${\cal R}_1$, and ${\cal R}_2$ respectively, ${\mathbf H}_{01}$, ${\mathbf H}_{02}$, ${\mathbf H}_{13}$, and ${\mathbf H}_{23}$ are the real $n \times n$ MIMO channel matrices corresponding to the ${\cal S}$-${\cal R}_1$, ${\cal S}$-${\cal R}_2$, ${\cal R}_1$-${\cal D}$, and ${\cal R}_2$-${\cal D}$ channels, and ${\mathbf{z}_1}$, ${\mathbf{z}_2}$, and ${\mathbf{z}_3}$ are the $n \times 1$ Gaussian noise vectors with distribution ${\cal N}(\mathbf{0}, \mathbf{I})$ at ${\cal R}_1$, ${\cal R}_2$, and ${\cal D}$, respectively.

\begin{figure}[ht]
\begin{center}
\begin{tikzpicture}[scale=0.55]
\tikzset{vertex/.style = {shape=circle,very thick,draw,minimum size=2mm}}
\tikzset{edge/.style = {->,> = latex', very thick}}

\node[vertex] (0) at (0,0) {0};
\node[vertex] (1) at (2,1.5) {1};
\node[vertex] (2) at (2,-1.5) {2};
\node[vertex] (3) at (4,0) {3};

\path[->,>=stealth,draw,thick]
 (0) edge node [above left] {${\mathbf H}_{01}$} (1)
 (2) edge node [below right] {${\mathbf H}_{23}$} (3);

\node[right] at (1,-2.5) {\textbf{State 1}};

\node[vertex] (0) at (6,0) {0};
\node[vertex] (1) at (8,1.5) {1};
\node[vertex] (2) at (8,-1.5) {2};
\node[vertex] (3) at (10,0) {3};

\path[->,>=stealth,draw,thick]
 (0) edge node [below left] {${\mathbf H}_{02}$} (2)
 (1) edge node [above right] {${\mathbf H}_{13}$} (3);

\node[right] at (7,-2.5) {\textbf{State 2}};


\node[vertex] (0) at (0,-6) {0};
\node[vertex] (1) at (2,-4.5) {1};
\node[vertex] (2) at (2,-7.5) {2};
\node[vertex] (3) at (4,-6) {3};

\path[->,>=stealth,draw,thick]
 (1) edge node [above right] {${\mathbf H}_{13}$} (3)
 (2) edge node [below right] {${\mathbf H}_{23}$} (3);

\node[right] at (1,-8.5) {\textbf{State 3}};


\node[vertex] (0) at (6,-6) {0};
\node[vertex] (1) at (8,-4.5) {1};
\node[vertex] (2) at (8,-7.5) {2};
\node[vertex] (3) at (10,-6) {3};

\path[->,>=stealth,draw,thick]
 (0) edge node [above left] {${\mathbf H}_{01}$} (1)
 (0) edge node [below left] {${\mathbf H}_{02}$} (2);

\node[right] at (7,-8.5) {\textbf{State 4}};

\draw[dashed] (5,-9)-- (5,3);
\draw[dashed] (-1,-3.2)-- (11,-3.2);

\end{tikzpicture}
\caption{States of the diamond channel}
\label{states}
\end{center}
\end{figure}
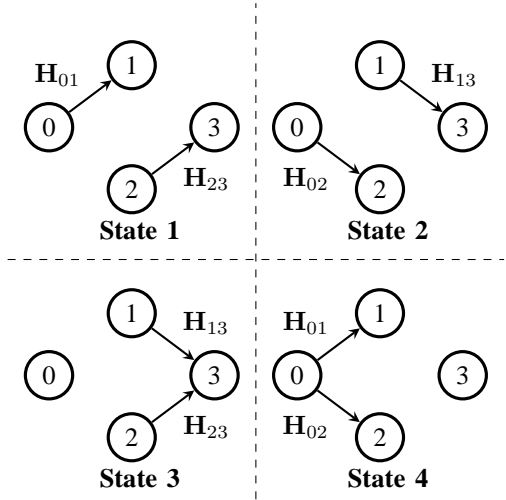
The nodes are half-duplex. The four possible relaying states are shown in Fig. \ref{states}.
As in \cite{BagMotKha14}, we assume constant power constraints for each node across all states. For nodes ${\cal S}$, ${\cal R}_1$, and ${\cal R}_2$, without loss of generality, the power constraints are taken to be 1, i.e., $P_0=P_1=P_2=1$. Let 
\[
{C}({\mathbf H}, P) = \max\limits_{{\mathbf Q} \succeq 0, \tr({\mathbf Q}) \leq P}\frac{1}{2}\log\mbox{det}({\mathbf I} + {\mathbf H} {\mathbf Q} {\mathbf H}^T).
\]
The channel parameters are defined as follows: $C_{01} = {C}({\mathbf H}_{01}, 1)$, $C_{02} = {C}({\mathbf H}_{02}, 1)$, $C_{13} = {C}({\mathbf H}_{13}, 1)$, $C_{23} = {C}({\mathbf H}_{23}, 1)$, $C_{012} = {C}({\mathbf H}_{012}, 1)$, and $C_{123} = {C}({\mathbf H}_{123}, 2)$,  where ${\mathbf H}_{012}^T=[ {\mathbf H}_{01}^T \; {\mathbf H}_{02}^T ]$ and ${\mathbf H}_{123}=[ {\mathbf H}_{13} \; {\mathbf H}_{23} ]$. The optimal covariance matrix ${\mathbf Q}$ corresponding to each of these capacities are denoted ${\mathbf K}_{01}$, ${\mathbf K}_{02}$, ${\mathbf K}_{13}$, ${\mathbf K}_{23}$, ${\mathbf K}_{012}$, and ${\mathbf K}_{123}$, respectively. For example, we have 
\[
C_{13} = \frac{1}{2}\log\mbox{det}({\mathbf I} + {\mathbf H}_{13} {\mathbf K}_{13} {\mathbf H}_{23}^T).
\]

\subsection{MDF-MAC protocol}
The MDF-MAC protocol is a multihopping decode-and-forward protocol using States 1, 2 and 3.  The total transmission time is normalized to 1 and States 1, 2, and 3 are used for $t_1$, $t_2$, and $t_3$ fractions of the total transmission time. Let $R_1$ and  $R_2$ be the rates of transmission from relays ${\cal R}_1$ and ${\cal R}_2$ to the destination in the multiple access state (State 3). Then, the maximum achievable rate $R_{\text{MAC}}$ from ${\cal S}$ to ${\cal D}$ of the MDF-MAC scheme is given by
\begin{align}
R_{\text{MAC}} = \max_{\sum_i t_i =1, t_i \ge 0} & \left\{ \min\{ t_1 C_{01}, t_2 C_{13} + R_1\} \right. \nonumber \\ 
&\left. + \min \{ t_2 C_{02}, t_1 C_{23} + R_2\} \right\}.
\label{Rmacmaxmin}
\end{align}

\section{Gap from Capacity of the MDF-MAC protocol for $\Delta > 0$}
In this section, we show that the MDF-MAC is within constant gap of capacity for $\Delta > 0$ for the 2-relay MIMO Gaussian diamond channel. In \cite{BagMotKha14}, the MDF-MAC was shown to be within 0.71 bits of capacity for the {\em single-antenna} Gaussian diamond channel. We will first summarize the steps in the analysis in \cite{BagMotKha14}, discuss the main difficulties in extending this result to the MIMO case, and then present our solution that overcomes these difficulties.

In \cite{BagMotKha14}, the gap result is obtained as follows.
\begin{itemize}
\item[Step 1.] A capacity upper bound is obtained by considering the dual of the linear program associated with the HD cut-set bound.
\item[Step 2.] A linear program is formulated for determining the achievable rate $R_{\text{MAC}}$ using the MDF-MAC protocol.
\item[Step 3.] An achievable $R_{\text{MAC}}$ is identified from the formulated linear program.
\item[Step 4.] The gap between the achievable rate and the upper bound is shown to be bounded if $C_{123} - C_{\text{MAC}}$ and $C_{123} - (C_{13} + C_{23})$ are both bounded by a finite constant, where $C_{\text{MAC}}$ is the sum rate in the MAC state of MDF-MAC.
\item[Step 5.] $C_{123} - C_{\text{MAC}}$ and $C_{123} - (C_{13} + C_{23})$ are shown to be bounded by finite constants.
\end{itemize}

The main difficulties in obtaining a similar gap result for the MIMO case are in Steps 2 and 5. In step 2, a linear program is easily obtained in the single-antenna case since the rate region for $(R_1, R_2)$ in the MAC state is a pentagon specified by a finite number of linear inequalities. In the MIMO setting, the rate region for $(R_1, R_2)$ is an infinite  union of such pentagons and cannot be exactly described by a finite number of linear inequalities. This capacity region for the MAC state is given by \cite{GolJafJinVis03}
\[
{\cal C}_{\text{MAC}} = \bigcup_{{\mathbf Q}_1, {\mathbf Q}_2,: {\mathbf Q}_i\succeq 0, \tr({\mathbf Q_i}) \leq 1} {\cal C}'_{\text{MAC}} ({\mathbf Q}_1, {\mathbf Q}_2),
\]
where ${\cal C}'_{\text{MAC}}({\mathbf Q}_1, {\mathbf Q}_2)$ is a pentagon obtained by choosing the covariance matrices at the relays ${\cal R}_1$ and ${\cal R}_2$ to be ${\mathbf Q}_1$ and ${\mathbf Q}_2$, respectively, i.e., ${\cal C}'_{\text{MAC}}({\mathbf Q}_1, {\mathbf Q}_2)$ is the set of all $(R_1, R_2)$ satisfying
\allowdisplaybreaks{\begin{align}
R_1 & \le \frac{1}{2}\log\mbox{det}({\mathbf I} + {\mathbf H}_{13} {\mathbf Q}_1 {\mathbf H}_{13}^T) \triangleq C_{13}^\prime \nonumber\\
R_2 & \le \frac{1}{2}\log\mbox{det}({\mathbf I} + {\mathbf H}_{23} {\mathbf Q}_2 {\mathbf H}_{23}^T) \triangleq C_{23}^\prime \nonumber\\
R_1 + R_2 & \le \frac{1}{2}\log\mbox{det}({\mathbf I} + {\mathbf H}_{13} {\mathbf Q}_1 {\mathbf H}_{13}^T + {\mathbf H}_{23} {\mathbf Q}_2 {\mathbf H}_{23}^T) \nonumber \\
& \triangleq C_{\text{MAC}}^\prime. 
\label{Cprimes}
\end{align}}
By fixing ${\mathbf Q}_1$ and ${\mathbf Q}_2$, we can get a linear program for $R_{\text{MAC}}$ even in the MIMO case. However, this should be done carefully since we need to be able to bound $C_{123} - C_{\text{MAC}}$ in Step 5. Thus, this choice affects Steps 3, 4, and 5. We study this problem and provide such an appropriate choice for ${\mathbf Q}_1$ and ${\mathbf Q}_2$.  

Our analysis is presented in detail in the following subsections.

\subsection{Upper bound (Step 1)}
The upper bound from \cite[Eqn. (37)]{BagMotKha14} is valid for our MIMO setting as well and is given by:
\small{
\begin{align}
R_{up}^1 & = \frac{C_{01}(C_{02} + C_{13})}{C_{01} + C_{13}} + \frac{-C_{02}\Delta}{(C_{123} - C_{13} + C_{02})(C_{01} + C_{13})} + \delta\\
R_{up}^2 &= \frac{C_{02}(C_{01} + C_{23})}{C_{02} + C_{23}} + \frac{-C_{01}\Delta}{(C_{123} - C_{23} + C_{01})(C_{02} + C_{23})} + \delta, 
\end{align}}
for $\Gamma' \le 0$ and $\Gamma' > 0$, respectively, where $\Gamma' = C_{02}[C_{123} - C_{23}] - C_{01}[C_{123} - C_{23}]$, and $\delta = \max(C_{123} - (C_{13} + C_{23}), 0)$.

\subsection{Formulating the linear program for $R_{\text{MAC}}$: Choice of ${\mathbf Q}_1$ and ${\mathbf Q}_2$ (Steps 2-4)}

For a given ${\mathbf Q}_1$ and ${\mathbf Q}_2$, we defined $C^\prime_{13}, C^\prime_{23}$ and $C^\prime_{\text{MAC}}$ in (\ref{Cprimes}). For this ${\mathbf Q}_1$ and ${\mathbf Q}_2$, we can formulate a linear program for $R_{\text{MAC}}$ as follows using (\ref{Rmacmaxmin}) and (\ref{Cprimes}).
 \begin{IEEEeqnarray}{lll}
\text{maximize} & \; & R_{\text{MAC}}
 \nonumber \\
\text{subject to}: & \; & R_{\text{MAC}} \leq t_1C_{01}+t_2C_{02}
 \nonumber \\
 & \; & R_{\text{MAC}} \leq t_2(C_{02}+C_{13}) + R_1
 \nonumber \\
 & \; & R_{\text{MAC}} \leq t_1(C_{01}+C_{23}) + R_2
 \nonumber \\
 & \; & R_{\text{MAC}} \leq t_1C_{23}+t_2C_{13} + R_1 + R_2
 \nonumber \\
 & \; & R_1 \leq t_3 C_{13}^\prime
 \nonumber \\
 & \; & R_2 \leq t_3 C_{23}^\prime
  \nonumber \\
 & \; & R_1 + R_2 \leq t_3 C_{\text{MAC}}^\prime
   \nonumber \\
 & \; & \sum_{i=1}^3 t_i=1,t_i \leq 0.
   \nonumber 
\end{IEEEeqnarray}
Using Fourier-Motzkin elimination to eliminate variables $R_1$ and $R_2$, the above optimization problem can be reduced to:
 \begin{IEEEeqnarray}{lll}
\text{maximize }& \; & R_{\text{MAC}}
 \nonumber \\
\text{subject to}: & \; & R_{\text{MAC}} \leq t_1C_{01}+t_2C_{02} \label{eq:op1}\\
 & \; & R_{\text{MAC}} \leq t_2(C_{02}+C_{13}) + t_3 C_{13}^\prime \label{eq:op2} \\
 & \; & R_{\text{MAC}} \leq t_1(C_{01}+C_{23}) + t_3 C_{23}^\prime \label{eq:op3}\\
 & \; & R_{\text{MAC}} \leq t_1C_{23}+t_2C_{13} + t_3 C_{\text{MAC}}^\prime \label{eq:op4}\\
 & \; & \sum_{i=1}^3 t_i=1,t_i \geq 0.
\end{IEEEeqnarray}
Now, we will choose a feasible rate for this linear program such that the gap from the upper bound can be bounded later. Setting to equality the inequalities \eqref{eq:op1}, \eqref{eq:op2}, \eqref{eq:op4} gives:
\begin{IEEEeqnarray}{lll}
t_1 =   \frac{C_{13}(C_{\text{MAC}}^\prime - C_{13})+C_{13}^\prime C_{02}}{C_{den}} \, , && \nonumber \\
t_2 =   \frac{C_{01}(C_{\text{MAC}}^\prime - C_{13}^\prime)+C_{13}^\prime C_{23}}{C_{den}}\, , && \label{soln1} \\
t_3 =   \frac{\Delta}{C_{den}}\, , && \nonumber \\
R_{\text{MAC-MDF}}^{1} =\frac{C_{01}(C_{02} - C_{13}) C_{\text{MAC}}^\prime  -C_{01} C_{13} C_{13}^\prime + C_{02}C_{13}^\prime C_{23}}{C_{den}},&& \nonumber
\end{IEEEeqnarray}
where the denominator $C_{den} = (C_{\text{MAC}}^\prime-C_{13}^\prime (C_{01}+C_{13})+C_{02}(C_{01}+C_{13}^\prime)+(C_{13}^\prime - C_{13})C_{23} $. Note that the equivalent expressions in \cite{BagMotKha14} are obtained as a special case of the above by setting $C_{13}^\prime=C_{13}$, $C_{23}^\prime=C_{23}$ and $C_{\text{MAC}}^\prime=C_{\text{MAC}}$. However, we are going to make a different choice here. 

The solution in (\ref{soln1}) would be feasible for the linear program if the inequality \eqref{eq:op3} is also satisfied. Inequality \eqref{eq:op3} is satisfied if
\[
t_1C_{01}+t_2C_{02} \le t_1(C_{01}+C_{23}) + t_3 C_{23}^\prime,
\]
i.e., if $t_2C_{02} \le t_1 C_{23} + t_3 C_{23}^\prime$. Substituting for $t_1$, $t_2$, and $t_3$ from the solution above, and simplifying, we can rewrite this condition as
\[
\Delta[C_{13}^\prime + C_{23}^\prime - C_{\text{MAC}}^\prime] \ge C_{01}C_{02}(C_{13} - C_{13}^\prime).
\]
Note that $C_{13}^\prime + C_{23}^\prime - C_{\text{MAC}}^\prime \ge 0$. Therefore, the required condition is satisfied for $\Delta > 0$ if we choose $C_{13}^\prime=C_{13}$, i.e., ${\mathbf Q}_1 = {\mathbf K}_{13}$. If we also choose $C_{23}^\prime=C_{23}$, i.e., ${\mathbf Q}_2 = {\mathbf K}_{23}$, we will not be able to bound $C_{123} - C'_{\text{MAC}}$ later. Therefore, we will make a different choice for ${\mathbf Q}_2$. We choose ${\mathbf Q}_2$ to be the covariance matrix obtained using the water-filling algorithm treating the noise covariance matrix to be $I+{\mathbf H}_{13} {\mathbf K}_{13} {\mathbf H}_{13}^T$. This is the solution at the end of the first iteration of the iterative waterfilling algorithm in \cite{YuRheBoyCio04}. Thus, we choose ${\mathbf Q}_2$ to be ${\mathbf K}'_{23}$ given by:
\[
{\mathbf K}_{23}^\prime = \mbox{arg} \max\limits_{{\mathbf Q}\succeq 0, \tr({\mathbf Q}) \leq 1}\frac{1}{2}\log\mbox{det}({\mathbf I}+{\mathbf H}_{13} {\mathbf K}_{13} {\mathbf H}_{13}^T+ {\mathbf H}_{23} {\mathbf Q} {\mathbf H}_{23}^T),
\]
resulting in
\begin{IEEEeqnarray}{lll}
C^\prime_{23}= \frac{1}{2}\log\mbox{det}\big({\mathbf I} + {\mathbf H}_{23}{\mathbf K}_{23}^\prime {\mathbf H}_{23}^T\big),&&
\nonumber \\
C_{\text{MAC}}^\prime=C'_{\text{MAC1}} = \frac{1}{2}\log\mbox{det}\big({\mathbf I} + {\mathbf H}_{13} {\mathbf K}_{13} {\mathbf H}_{13}^T+{\mathbf H}_{23} {\mathbf K}_{23}^\prime {\mathbf H}_{23}^T\big)&&.
\nonumber 
\end{IEEEeqnarray}
Thus, for $\Gamma' \le 0$, we choose the pentagon $\mathcal{C}_{\text{MAC}}({\mathbf K}_{13}, {\mathbf K}'_{23})$ given by:
\[
\{ (R_1,R_2):R_1 \leq  C_{13},R_2 \leq  C^\prime_{23},R_1+R_2 \leq  C'_{\text{MAC1}} \}. \]
Similarly, for $\Gamma' > 0$, we set to equality the inequalities \eqref{eq:op1}, \eqref{eq:op3}, \eqref{eq:op4} to get:
\[
R_{\text{MAC-MDF}}^{2} =   \frac{C_{02}(C_{01} - C_{23}) C_{\text{MAC}}^\prime  -C_{02} C_{23} C_{23}^\prime + C_{01}C_{13} C_{23}^\prime}{C_{den2}}, 
\]
where $C_{den2} = (C_{\text{MAC}}^\prime-C_{23}^\prime)(C_{02}+C_{23})+C_{01}(C_{02}+C_{23}^\prime)+(C_{23}^\prime - C_{23})C_{13}$. Then, we choose $C_{23}^\prime=C_{23}$, i.e., ${\mathbf Q}_2 = {\mathbf K}_{23}$ and 
${\mathbf Q}_1$ to be ${\mathbf K}'_{13}$ given by:
\[
{\mathbf K}_{13}^\prime = \mbox{arg} \max\limits_{{\mathbf Q}\succeq 0, \tr({\mathbf Q}) \leq 1}\frac{1}{2}\log\mbox{det}({\mathbf I}+{\mathbf H}_{13} {\mathbf Q} {\mathbf H}_{13}^T+ {\mathbf H}_{23} {\mathbf K}_{23} {\mathbf H}_{23}^T),
\]
resulting in
\begin{IEEEeqnarray}{lll}
C^\prime_{13}= \frac{1}{2}\log\mbox{det}\big({\mathbf I} + {\mathbf H}_{13}{\mathbf K}_{13}^\prime {\mathbf H}_{13}^T\big),&&
\nonumber \\
C_{\text{MAC}}^\prime=C'_{\text{MAC2}} = \frac{1}{2}\log\mbox{det}\big({\mathbf I} + {\mathbf H}_{13} {\mathbf K}'_{13} {\mathbf H}_{13}^T+{\mathbf H}_{23} {\mathbf K}_{23} {\mathbf H}_{23}^T\big)&&.
\nonumber 
\end{IEEEeqnarray}
Thus, for $\Gamma' > 0$, we choose the pentagon $\mathcal{C}_{\text{MAC}}({\mathbf K}'_{13}, {\mathbf K}_{23})$ given by:
\[
\{ (R_1,R_2):R_1 \leq  C'_{13},R_2 \leq  C_{23},R_1+R_2 \leq  C'_{\text{MAC2}} \}. \]

It is worth noting that the gap between the sum rate achieved after one iteration of iterative waterfilling $C'_{\text{MAC1}}$ (or $C'_{\text{MAC2}}$) and the sum capacity of the MIMO MAC can be bounded by a finite constant \cite{YuRheBoyCio04}.

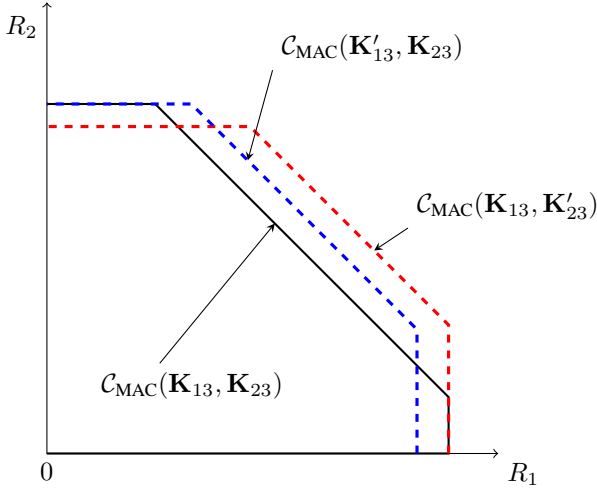
\begin{figure}[ht]
\begin{tikzpicture}[scale=0.3]

\draw [<->] (0,20) -- (0,0) -- (20,0);

\node[below right] at (20,0) {$R_1$};
\node[above left] at (0,18) {$R_2$};
\node[below] at (0,0) {$0$};

\node[above right] at (2,2) {${\cal C}_{\text{MAC}}({\mathbf K}_{13}, {\mathbf K}_{23})$};
\draw[->,>=stealth,black] (5,4) -- (10.1,10.2);

\draw[black, thick] (0,0)--(17.8,0)--(17.8,2.5)--(4.8,15.5)--(0,15.5);
\draw[red, very thick, dashed] (17.8,0)--(17.8,5.7)--(9,14.5)--(0,14.5);
\draw[blue, very thick, dashed] (16.4,0)--(16.4,5.5)--(6.4,15.5)--(0,15.5);

\node[above right] at (10,17) {${\cal C}_{\text{MAC}}({\mathbf K}'_{13}, {\mathbf K}_{23})$};
\draw[->,>=stealth,black] (10,17) -- (8.9,13);

\node[above right] at (16,10) {${\cal C}_{\text{MAC}}({\mathbf K}_{13}, {\mathbf K}'_{23})$};
\draw[->,>=stealth,black] (16,10) -- (14.5,9);

\end{tikzpicture}
\caption{An illustration of the rate regions for different choices of ${\mathbf Q}_1$,  ${\mathbf Q}_2$}
\label{macregions}
\end{figure}
Fig. \ref{macregions} gives an illustration of the rate regions corresponding to the chosen covariance matrices. Three regions are shown: ${\cal C}_{\text{MAC}}({\mathbf K}_{13}, {\mathbf K}_{23})$, ${\cal C}_{\text{MAC}}({\mathbf K}'_{13}, {\mathbf K}_{23})$ (our choice for $\Gamma' > 0$), and ${\cal C}_{\text{MAC}}({\mathbf K}_{13}, {\mathbf K}'_{23})$ (our choice for $\Gamma' \le 0$).

Now, we can show the difference between the achievable rate and the upper bound to be:
\begin{align}
&\kappa_{\text{MAC}}^1 \triangleq  R_{up}^1 - R_{\text{MAC-MDF}}^{1} \nonumber \\
& = \frac{C_{02}(C_{123} - C'_{\text{MAC1}}) \Delta}{(C_{01}+C_{13})(C'_{\text{MAC1}} - C_{13} + C_{02})(C_{123} - C_{13} + C_{02})} + \delta, 
\label{kappa1}
\end{align}
or
\begin{align}
&\kappa_{\text{MAC}}^2 \triangleq  R_{up}^2 - R_{\text{MAC-MDF}}^{2} \nonumber \\
& = \frac{C_{01}(C_{123} - C'_{\text{MAC2}}) \Delta}{(C_{02}+C_{23})(C'_{\text{MAC2}} - C_{23} + C_{01})(C_{123} - C_{23} + C_{01})} + \delta, 
\label{kappa2}
\end{align}
depending on whether $C'_{\text{MAC}} = C'_{\text{MAC1}}$ or $C'_{\text{MAC2}}$. This gap can be bounded if we can bound $C_{123} - C'_{\text{MAC}}$ and $\delta$. We will do this in the next two subsections.

\subsection{Bounding $C_{123} - C'_{\text{MAC}}$ (Step 5)}

Let $C_{sum}$ be the sum-rate capacity of MAC channel between the relays and the destination. From theorem 3 in  \cite{YuRheBoyCio04}, we have the following result on the convergence of the iterative waterfilling algorithm after the first iteration.
\begin{IEEEeqnarray}{rll}
C_{sum}-C'_{\text{MAC1}}&  \leq   \frac{n}{2}, \nonumber \\
C_{sum}-C'_{\text{MAC2}}&  \leq   \frac{n}{2}.  \nonumber
\end{IEEEeqnarray}
Equivalently, we have 
\begin{equation}
C_{sum}-C'_{\text{MAC}} \leq \frac{n}{2}.
\label{eqmac}
\end{equation}
\begin{lemma}
\label{lemmadiff}
$C_{123}-C'_{\text{MAC}}<\frac{n}{2}\log 4n$.
\end{lemma}
\begin{proof}
\begin{IEEEeqnarray}{rll}
C_{123}&=&\max\limits_{{\mathbf Q}\succeq 0, \tr({\mathbf Q}) \leq 2}\frac{1}{2}\log\mbox{det}({\mathbf I} + {\mathbf H}_{123} {\mathbf Q} {\mathbf H}_{123}^T) \nonumber\\
&\stackrel{(a)}{\leq} & \frac{1}{2} \log\mbox{det}({\mathbf I} + {\mathbf H}_{123}(2{\mathbf I}) {\mathbf H}_{123}^T) \nonumber\\
&=& \frac{1}{2} \log\mbox{det}({\mathbf I}+2 {\mathbf A}),
\label{c123bound}
\end{IEEEeqnarray}
where (a) is true because $2{\mathbf I} \succeq {\mathbf Q}$ and $\log \det(\cdot)$ is increasing on the cone of positive-definite Hermitian matrices, and ${\mathbf A} \triangleq {\mathbf H}_{123}{\mathbf H}_{123}^T$. Now, note that $C_{sum}$ is at least as much as the sum rate achieved by ${\mathbf Q}_1 = {\mathbf Q}_2 = \frac{1}{n}{\mathbf I}$, i.e., we have
\begin{IEEEeqnarray}{rll}
C_{sum}&\geq &\frac{1}{2}\log\mbox{det}\big({\mathbf I}+\frac{1}{n}{\mathbf H}_{13} {\mathbf H}_{13}^T+\frac{1}{n} {\mathbf H}_{23} {\mathbf H}_{23}^T\big) 
\nonumber \\
 &  = & \frac{1}{2}\log\mbox{det}\big({\mathbf I}+\frac{1}{n}{\mathbf H}_{123} {\mathbf H}_{123}^T\big) 
\nonumber \\
 &  = & \frac{1}{2}\log\mbox{det}\big({\mathbf I}+\frac{1}{n}{\mathbf A}\big).
\label{csumbound} 
\end{IEEEeqnarray}
Therefore, from (\ref{c123bound}) and (\ref{csumbound}), we have
\begin{IEEEeqnarray}{rll}
C_{123}-C_{sum}&\leq & \frac{1}{2} \log\mbox{det}({\mathbf I}+2{\mathbf A}) - \frac{1}{2}\log\mbox{det}\big({\mathbf I}+\frac{1}{n}{\mathbf A}\big) \nonumber\\
& =  & \frac{1}{2} \log\frac{\mbox{det}({\mathbf I}+2{\mathbf A})}{\mbox{det}({\mathbf I}+\frac{1}{n}{\mathbf A})}
\nonumber \\
& \stackrel{(a)}{=}& \frac{1}{2} \log \left( \prod_{i=1}^n \frac{1+2a_i}{1+\frac{a_i}{n}}\right) \nonumber \\
& \stackrel{(b)}{\leq}& \frac{1}{2} \log \left( \prod_{i=1}^n 2n \right) \nonumber\\
& =& \frac{n}{2}\log 2n,
\label{equnif}
\end{IEEEeqnarray}
where $a_1, a_2, \ldots, a_n$ are the eigen values of ${\mathbf A}$ in (a), and (b) is obtained by using Lemma \ref{prop1} in Appendix for the $y=0$ case.\\
Finally, from \eqref{eqmac} and \eqref{equnif}, the bound in the lemma is obtained.
\end{proof}

\subsection{Bounding $\delta$ (Step 5)}

\begin{lemma}
\label{lemmadelta}
$\delta \leq \frac{n}{2} \log 2n$.
\end{lemma}
\begin{proof}
Let $C_{123} > C_{13}+C_{23}$  (otherwise, $\delta = 0$) . Then 
\begin{IEEEeqnarray}{rll}
\delta & =   & C_{123} -( C_{13}+C_{23})
\nonumber \\
 & =  & \max\limits_{{\mathbf Q}\succeq 0, \tr({\mathbf Q}) \leq 2}\frac{1}{2}\log\mbox{det}({\mathbf I}+ {\mathbf H}_{123}{\mathbf Q}{\mathbf H}_{123}^T) \nonumber \\
 &&- \max\limits_{{\mathbf Q}\succeq 0, \tr({\mathbf Q}) \leq 1}\frac{1}{2}\log\mbox{det}({\mathbf I}+{\mathbf H}_{13}{\mathbf Q}{\mathbf H}_{13}^T)
 \nonumber \\
 &&- \max\limits_{{\mathbf Q}\succeq 0, \tr({\mathbf Q}) \leq 1}\frac{1}{2}\log\mbox{det}({\mathbf I}+ {\mathbf H}_{23}{\mathbf Q}{\mathbf H}_{23}^T) \nonumber \\
 & \leq & \frac{1}{2}\log\mbox{det}({\mathbf I}+2{\mathbf H}_{123}{\mathbf H}_{123}^T) \nonumber \\
&&-\frac{1}{2}\log\mbox{det}({\mathbf I}+\frac{1}{n}{\mathbf H}_{13}{\mathbf H}_{13}^T) \nonumber \\
&&- \frac{1}{2}\log\mbox{det}({\mathbf I}+\frac{1}{n}{\mathbf H}_{23}{\mathbf H}_{23}^T) \nonumber
\end{IEEEeqnarray}
Let ${\mathbf A}={\mathbf H}_{13}{\mathbf H}_{13}^T$ and ${\mathbf B}={\mathbf H}_{23}{\mathbf H}_{23}^T$. Since ${\mathbf H}_{123}=[{\mathbf H}_{13} \; {\mathbf H}_{23}]$, $\delta$ can be upper bounded as
\begin{equation}
\delta \leq  \frac{1}{2}\log  \frac{\mbox{det}({\mathbf I}+2{\mathbf A}+2{\mathbf B})}{\mbox{det}({\mathbf I}+\frac{1}{n}{\mathbf A}) \, \mbox{det}({\mathbf I}+\frac{1}{n}{\mathbf B})}.
\label{eq:delta1}
\end{equation}
Consider the term inside the log. We can upper bound it as follows using Lemmas \ref{fiedler} and \ref{prop1} in the Appendix. Let $\alpha_1 \geq \alpha_2 \geq \ldots \geq \alpha_n$ be the eigen values of ${\mathbf A}$ and $\beta_1 \geq \beta_2 \geq \ldots \geq \beta_n$ be the eigen values of ${\mathbf B}$.
\begin{IEEEeqnarray}{rll}
 \frac{\mbox{det}({\mathbf I}+2{\mathbf A}+2{\mathbf B})}{\mbox{det}({\mathbf I}+\frac{1}{n}{\mathbf A}) \, \mbox{det}({\mathbf I}+\frac{1}{n}{\mathbf B})} & = & \frac{\mbox{det}({\mathbf I}+2{\mathbf A}+2{\mathbf B})}{\prod_{i=1}^n \left(1+\frac{\alpha_i}{n}\right)\left(1+\frac{\beta_i}{n} \right)} \nonumber\\
& \stackrel{(a)}{\leq}& \frac{\prod_{i=1}^n (1+ 2 \alpha_i + 2 \beta_{n+1-i})}{\prod_{i=1}^n \left(1+\frac{\alpha_i}{n}\right)\left(1+\frac{\beta_i}{n} \right)} \nonumber\\
& = & \frac{\prod_{i=1}^n (1+ 2 \alpha_i + 2 \beta_{n+1-i})}{\prod_{i=1}^n \left(1+\frac{\alpha_i}{n}\right)\left(1+\frac{1}{n} \beta_{n+1-i} \right)} \nonumber\\
& = & \prod_{i=1}^n \frac{ (1+ 2 \alpha_i + 2 \beta_{n+1-i})}{\left(1+\frac{\alpha_i}{n}\right)\left(1+\frac{1}{n} \beta_{n+1-i} \right)} \nonumber\\
& \stackrel{(b)}{\leq} & \prod_{i=1}^n 2n  = (2n)^n,
\label{eq:delta2}
\end{IEEEeqnarray}
where (a) is obtained using Lemma \ref{fiedler}, and (b) is obtained using Lemma  \ref{prop1}.

Finally, using \eqref{eq:delta2} in \eqref{eq:delta1}, the required bound in the lemma is obtained.
\end{proof}

\subsection{Main Result: MDF-MAC achieves rates within constant gap of capacity for $\Delta > 0$}
\begin{theorem}
For the 2-relay MIMO Gaussian diamond channel with $n$ antennas per node, the multi-hopping decode-and-forward protocol MDF-MAC achieves rates within $n \log_2 (\sqrt{8}n)$ bits of capacity for $\Delta > 0$.
\end{theorem}
\begin{proof}
Using Lemma \ref{lemmadiff}, Lemma \ref{lemmadelta}, and the facts- $C'_{\text{MAC1}} \geq C_{13}$ and $C'_{\text{MAC2}} \geq C_{23}$, the gaps in (\ref{kappa1}) and (\ref{kappa2}) can be bounded as:
\[
\kappa_{\text{MAC}}^1 \le \frac{n}{2} \log 4n + \frac{n}{2} \log 2n = n \log(\sqrt{8}n),
\]
\[
\kappa_{\text{MAC}}^2 \le \frac{n}{2} \log 4n + \frac{n}{2} \log 2n = n \log(\sqrt{8}n),
\]
to get the required result.
\end{proof}

\section{Summary and future directions}
In this paper, we showed that for a half-duplex 2-relay MIMO Gaussian diamond channel a simple multihopping decode-and-forward protocol, MDF-MAC, achieves rates within a constant gap from capacity for $\Delta > 0$. The constant gap bound is $n \log_2 (\sqrt{8}n)$ bits, where $n$ is the number of antennas in each node. This is a generalization of the single-antenna diamond channel result in \cite{BagMotKha14}. We also identify the transmit covariance matrices to be used by each relay in the multiple-access (MAC) state of the MDF-MAC protocol. 

For $\Delta < 0$, the MDF-BC protocol achieves rates within a constant gap from capacity for the single-antenna case \cite{BagMotKha14}. It is not yet known if this result also extends to the multi-antenna setting. A constant gap bound linear in $n$ has been obtained recently in \cite{CarTunKno15}, but requires noisy network coding to be used instead of the simple decode-and-forward approach used in our work.


\section*{Appendix}
\begin{lemma}
Let $f(x,y)=\frac{1+2x+2y}{\left( 1+\frac{x}{n} \right) \left( 1+\frac{y}{n}\right)}$. Then, for $n \geq 1,$ 
\begin{IEEEeqnarray}{rll}
\sup_{x \geq 0, y \geq 0}f(x,y) &=& 2n. \nonumber
\end{IEEEeqnarray}
\label{prop1}
\end{lemma}
\begin{proof}
It can be easily seen that for $n \geq 1, x \geq 0, y \geq 0$:
\begin {equation}
 2n - f(x,y)=2n - \frac{1+2x+2y}{\left( 1+\frac{x}{n}\right)\left( 1+\frac{y}{n}\right)}\geq 0. 
\label{eq:claim1}
 \end {equation}
Further, we also have
\begin {equation}
\lim_{x \rightarrow \infty} f(x,0) = 2n. 
\label{eq:claim2}
 \end {equation}
Therefore, using \eqref{eq:claim1} and \eqref{eq:claim2}, the required result is obtained.
\end{proof}

\begin{lemma}
[Theorem in \cite{Fie71}, (2)]
Let ${\mathbf C}$ and ${\mathbf D}$ be hermitian $n \times n$ matrices whose eigen values are $c_1 \geq c_2 \geq \ldots \geq c_n$ and $d_1 \geq d_2 \geq \ldots \geq d_n$, respectively. Then, if $c_n + d_n \ge 0$ (which is true if ${\mathbf C}$ and ${\mathbf D}$ are positive semidefinite), we have
\begin{IEEEeqnarray}{rll}
\det({\mathbf C} + {\mathbf D}) &\leq & \prod_{i=1}^n (c_i + d_{n+1-i}). \nonumber
\end{IEEEeqnarray}
\label{fiedler}
\end{lemma}

\bibliographystyle{IEEEtran}
\bibliography{references}
\end{document}